\pdfoutput=0
\documentclass[conference]{IEEEtran}
\usepackage{amsfonts}
\usepackage{graphicx}
\usepackage{color}
\usepackage{amsmath,amsfonts,amssymb,amsthm,epsfig,epstopdf,url,array}
\usepackage[framed,numbered,autolinebreaks,useliterate]{mcode}
\usepackage{url,textcomp}
\usepackage{authblk}
\usepackage{cite}

\newtheorem{theorem}{Theorem}
\newtheorem{lemma}{Lemma}

\begin{document}
\title{Outage Analysis on Type I HARQ over Time-Correlated Rayleigh Fading Channels}
\author {Zheng~Shi\thanks{The authors are with the Department of Electrical and Computer Engineering, University of Macau, Macau.}}
\author{Shaodan~Ma}
\author{Kam-Weng Tam}
\affil {Department of Electrical and Computer Engineering, University of Macau, Macau}
\maketitle
\begin{abstract}
This paper analyzes Type I hybrid automatic repeat request (HARQ) over time-correlated Rayleigh fading channels. Due to the presence of channel time correlation, the analysis is more challenging than the prior analysis in the literature. Outage probability is first derived as a weighted sum of joint CDF of multiple independent Gamma random variables based on an infinite series representation. A truncation method is proposed for efficient computation of the outage probability and it is proved that the truncation error decreases exponentially with the truncation order. Asymptotic outage probability is then derived in a simple form, with which the impacts of packet transmission rate, transmit power and channel time correlation could be decoupled and analyzed clearly. Based on the asymptotic outage probability, diversity order of HARQ is also analyzed. It is found that full diversity can be achieved even under time correlated fading channels and the time correlation of the channels has negative effect on the outage probability under high signal-to-noise ratio (SNR). Finally, our analytical results are validated by Monte-Carlo simulations.

\end{abstract}
\IEEEpeerreviewmaketitle
\section{Introduction}
Hybrid automatic repeat request (HARQ) is a reliable and powerful technique to combat the detrimental effect of fading and noisy channels by using the combination of forward error correction code (FEC) and automatic repeat request (ARQ). There are two kinds of HARQ schemes, i.e., Type I HARQ and HARQ with soft combining \cite{dahlman20134g}. For Type I HARQ, the erroneously received packets are discarded when retransmissions are requested, while for the other kind of HARQ, those erroneously received packets are stored in a buffer memory for joint decoding with subsequent received packets from retransmissions. Based on the decoding approaches adopted, HARQ with soft combining is further classified into two types, i.e., HARQ with chase combining (HARQ-CC) and HARQ with incremental redundancy (HARQ-IR). Although Type I HARQ performs worse than HARQ with soft combining, it has lower decoding complexity and requires less memory. It still has quite a lot of applications in practice. We thus put our focus on the analysis of Type I HARQ.

Most of the prior analyses in the literature consider either quasi-static fading channels (fully correlated fading channels) \cite{shen2009average} or independent fading channels \cite{caire2001throughput,zheng2005optimizing,chaitanya2015energy}. To be specific, in \cite{shen2009average}, packets in all HARQ rounds are assumed to experience an identical channel realization, and the average throughput of Type I HARQ is optimized through power allocation among retransmissions. Unlike \cite{shen2009average}, independent Rayleigh fading channels are considered in \cite{caire2001throughput}, and the throughput of Type I HARQ is analyzed. Considering the same independent fading channels as \cite{caire2001throughput}, a systematic approach with adaptive modulation and coding is proposed to maximize the throughput for Type I HARQ in \cite{zheng2005optimizing}. Moreover, MIMO systems with Type-I HARQ are investigated over independent fading channels in \cite{chaitanya2015energy}. An optimal power allocation solution is derived in closed-form to minimize the asymptotic outage probability given an average power constraint.

Except the above two types of fading channels, another general channel model is time correlated fading channel, which usually occurs when the transceiver has low-to-medium mobility \cite{kim2011optimal,jin2011optimal}. In \cite{kim2011optimal,jin2011optimal}, the time correlation among fading channels is considered in the analysis of HARQ-CC. The most fundamental metric, outage probability, is analyzed by using approximations. However, it is hard to extract meaningful insight of time correlation and other parameters from the complicated expression of the outage probability.

Considering the wide occurrence of time correlated fading channels, we analyze Type I HARQ over time-correlated Rayleigh fading channels in this paper. The fading channels are modeled as a multivariate Rayleigh distribution with exponential correlation. The outage probability is first derived as a weighted sum of joint CDF of multiple independent Gamma random variables (RVs) based on an infinite series representation. For efficient computation of the outage probability, a truncation method is proposed and it is proved that the truncation error decreases exponentially with the truncation order. To extract meaningful insight, asymptotic outage probability is then derived in a simple form, with which the impacts of packet transmission rate, transmit power and channel time correlation could be decoupled and analyzed clearly. The result of asymptotic outage probability also enables the analysis of diversity order. It is found that full diversity can be achieved even under time correlated fading channels and the time correlation of the channels has negative effect on the outage probability under high signal-to-noise ratio (SNR). Finally, our analytical results are validated by Monte-Carlo simulations. Our analysis could thus serve a solid foundation for system design and optimization.


The remainder of this paper is organized as follows. In Section \ref{sec:sys_mod}, the system model is introduced. Outage analysis is conducted and the asymptotic result is derived in Section \ref{sec:per_ana}. Numerical results are then presented for validations and discussions in Section \ref{sec:num}. Finally, Section \ref{sec:con} concludes this paper.

\section{System Model}\label{sec:sys_mod}
This paper considers a point-to-point Type I HARQ system operating over time-correlated Rayleigh fading channels. Each message comprises $b$ bits of information. For reliable transmission, the $b$ bits information is encoded by a channel encoder at a rate of $R$. Following Type I HARQ protocol, the same encoded packet is repetitively transmitted to the destination in multiple HARQ rounds until an acknowledgement (ACK) of successful reception is received from the destination or the maximum number of transmissions is reached. At the destination, signal detection is performed based on the received signal on current HARQ round. If the detection is failed, the destination would discard the received signal and feed back a negative acknowledge (NACK) to request a retransmission to the source.

Denote the encoded message with unit mean power as $x$. At the $k$th HARQ round, the encoded message is transmitted with a power of $P_k$ through a noisy time-correlated Rayleigh fading channel. The received signal at the destination in the $k$th HARQ round is then written as
\begin{equation}\label{eqn:rec_sig_des}
{y_k} = \sqrt {{P_k}} {h_k}x + {n_k},\quad 1 \le k \le K,
\end{equation}
where $K$ denotes the maximum number of transmissions, $n_k$ refers to complex Gaussian white noise with zero mean and unit variance, i.e., $n_k \sim \mathcal {CN} (0,1)$, and ${h_k}$ denotes Rayleigh fading channel coefficient. In this paper, time-correlated Rayleigh fading channel is considered and the channel coefficient $h_k$ is generally modeled as a multivariate Rayleigh distribution with exponential correlation, such that
\begin{equation}\label{eqn:mod_mul_ray_exp_corr}
{h_k} = {\rho ^{k + \delta  - 1}}\sigma_k{h_0} + \sqrt {1 - {\rho ^{2\left( {k + \delta  - 1} \right)}}} \sigma_k {w_k},\, 1 \le k \le K,\delta  > 0,
\end{equation}
where $\rho$ and $\delta$ denote the time correlation and the channel feedback delay, respectively, and $h_0, w_1, \cdots, w_K$ follow independent and identical complex Gaussian distributions with zero mean and unit variance, i.e., $h_0, w_1, \cdots w_K \sim \mathcal {CN} (0,1)$ and the average channel power gain $|h_k|^2$ is ${\rm E}(|h_k|^2)={\sigma_k}^2$. As derived in \cite{beaulieu2011novel}, the joint PDF of the channel amplitudes $|h_1|,\cdots,|h_K|$ is given by
\begin{multline}\label{eqn:joint_pdf_mul_rey_org}
{f_{\left| {{h_1}} \right|, \cdots ,\left| {{h_K}} \right|}}\left( {{x_1}, \cdots ,{x_K}} \right) = \int\limits_0^\infty  {{e^{ - t}}\prod\limits_{k = 1}^K {\frac{{{x_k}}}{{{\sigma _k}^2\left( {\frac{{1 - {\rho ^{2\left( {k + \delta  - 1} \right)}}}}{2}} \right)}}} } \times \\
{e^{ - \frac{{{x_k}^2 + {\sigma _k}^2{\rho ^{2\left( {k + \delta  - 1} \right)}}t}}{{{\sigma _k}^2\left( {1 - {\rho ^{2\left( {k + \delta  - 1} \right)}}} \right)}}}}{I_0}\left( {\frac{{{x_k}\sqrt {t{\sigma _k}^2{\rho ^{2\left( {k + \delta  - 1} \right)}}} }}{{{\sigma _k}^2\left( {\frac{{1 - {\rho ^{2\left( {k + \delta  - 1} \right)}}}}{2}} \right)}}} \right)dt,{\mkern 1mu} 0 \le \rho  < 1.
\end{multline}

At the destination, the received signal-to-noise ratio (SNR) in the $k$th HARQ round is therefore given as
\begin{equation}\label{eqn:SNR_k}
{\gamma _k} = {P_k}{\left| {{h_k}} \right|^2}.
\end{equation}
Due to the time correlation among the channel coefficients $h_k$, the received SNRs are correlated among multiple HARQ rounds, which then complicates the analysis as shown in the following section.


\section{Outage Analysis}\label{sec:per_ana}
With Type I HARQ, outage may still occur when the transmissions in all HARQ rounds fail. The outage probability can then be written from information theoretical perspective as \cite{caire2001throughput,shen2009average}
\begin{equation}\label{eqn:out_k_harq}
{\mathcal P_{out}}\left( K \right) = \Pr \left( {{I_1} < R, \cdots ,{I_K} < R} \right)
\end{equation}
where $I_k$ represents the mutual information in the $k$th HARQ round and is given as ${I_k} = {\log _2}\left( {1 + {\gamma _k}} \right)$. Substituting the definition of $I_k$ into (\ref{eqn:out_k_harq}), the outage probability can be further written as
\begin{align}\label{eqn:out_k_harq_fur}
{\mathcal P_{out}}\left( K \right) &= \Pr \left( {{\gamma _1} < {2^R} - 1, \cdots ,{\gamma _K} < {2^R} - 1} \right)  \notag \\
&= {F_{{\gamma _1}, \cdots ,{\gamma _K}}}\left( {{2^R} - 1, \cdots ,{2^R} - 1} \right)
\end{align}
where ${F_{{\gamma _1}, \cdots ,{\gamma _K}}}\left( \cdot \right)$ denotes the joint CDF of SNRs $\gamma_1,\cdots ,\gamma_K$.
Clearly, the joint CDF ${F_{{\gamma _1}, \cdots ,{\gamma _K}}}\left( \cdot \right)$ should be derived to obtain ${\mathcal P_{out}}\left( K \right)$.
\subsection{Series representation of ${\mathcal P_{out}}\left( K \right)$}
 To derive ${F_{{\gamma _1}, \cdots ,{\gamma _K}}}\left( \cdot \right)$, the joint PDF ${f_{{\gamma _1}, \cdots ,{\gamma _K}}}\left(  \cdot  \right)$ with respect to $\gamma_1,\cdots ,\gamma_K$ is first derived in the following lemma.
\begin{lemma}\label{lemm:pdf}
Given ${\gamma _k} = {P_k}{\left| {{h_k}} \right|^2}$, the joint PDF of SNRs $\gamma_1,\cdots ,\gamma_K$ is given by
\begin{multline}\label{eqn:joint_pdf_snrs}
{f_{{\gamma _1}, \cdots ,{\gamma _K}}}\left( {{r_1}, \cdots ,{r_K}} \right) = \int\limits_0^\infty  {{e^{ - t}}\prod\limits_{k = 1}^K {\frac{1}{{{P_k}{\sigma _k}^2\left( {1 - {\rho ^{2\left( {k + \delta  - 1} \right)}}} \right)}}} } \times \\
{e^{ - \frac{{{r_k} + {P_k}{\sigma _k}^2{\rho ^{2\left( {k + \delta  - 1} \right)}}t}}{{{P_k}{\sigma _k}^2\left( {1 - {\rho ^{2\left( {k + \delta  - 1} \right)}}} \right)}}}}_0{F_1}\left( {;1;\frac{{{\rho ^{2\left( {k + \delta  - 1} \right)}}t{r_k}}}{{{P_k}{\sigma _k}^2{{\left( {1 - {\rho ^{2\left( {k + \delta  - 1} \right)}}} \right)}^2}}}} \right)dt
\end{multline}
where ${}_0F_1(\cdot)$ denotes the confluent hypergeometric limit function.
\end{lemma}
\begin{proof}
By substituting the identity ${I_{m - 1}}\left( x \right) = \frac{{{{\left( {\frac{x}{2}} \right)}^{m - 1}}}}{{\Gamma \left( m \right)}}{}_0{F_1}\left( {;m;{{\left( {\frac{x}{2}} \right)}^2}} \right)$ \cite[Eq. 10.39.9]{olver2010nist} into (\ref{eqn:joint_pdf_mul_rey_org}), and then applying Jacobian transformation as ${\gamma _k} = {P_k}{\left| {{h_k}} \right|^2}$, the lemma holds after some algebraic manipulations.
\end{proof}
By using Lemma \ref{lemm:pdf}, the joint CDF of SNRs $\gamma_1,\cdots ,\gamma_K$ can then be derived in the following theorem.
\begin{theorem} \label{the:joint_CDF_mix}The CDF of ${F_{{\gamma _1}, \cdots ,{\gamma _K}}}\left( {{z_1}, \cdots ,{z_K}} \right) $ can be written as a weighted sum of joint CDF of $K$ independent Gamma RVs ${{\bf A}_{\bf{n}}}$ with parameters $(n_k+1,{{P_k}{\sigma _k}^2\left( {1 - {\rho ^{2\left( {k + \delta  - 1} \right)}}} \right)})$. More precisely,
\begin{equation}\label{eqn:joint_CDF_of_snrs_the}
{F_{{\gamma _1}, \cdots ,{\gamma _K}}}\left( {{z_1}, \cdots ,{z_K}} \right) = \sum\limits_{{n_{\rm{1}}}, \cdots ,{n_K} = 0}^\infty  {{W_{\bf{n}}}{F_{{{\bf A}_{\bf{n}}}}}{\left( {{z_1}, \cdots ,{z_K}} \right)}}
\end{equation}
where ${\bf{n}} = \left[ {{n_1}, \cdots ,{n_K}} \right]$, the coefficient ${W_{\bf{n}}}$ is given as
\begin{equation}\label{eqn:W_n_def}
{W_{\bf{n}}} = \frac{{\rm{1}}}{{1 + \sum\limits_{k = 1}^K {\frac{{{\rho ^{2\left( {k + \delta  - 1} \right)}}}}{{1 - {\rho ^{2\left( {k + \delta  - 1} \right)}}}}} }}\frac{{\left( {\sum\limits_{k = 1}^K {{n_k}} } \right)!}}{{\prod\limits_{k = 1}^K {{n_k}!} }}\prod\limits_{k = 1}^K {{{\left( {\frac{{\frac{{{\rho ^{2\left( {k + \delta  - 1} \right)}}}}{{1 - {\rho ^{2\left( {k + \delta  - 1} \right)}}}}}}{{1 + \sum\limits_{k = 1}^K {\frac{{{\rho ^{2\left( {k + \delta  - 1} \right)}}}}{{1 - {\rho ^{2\left( {k + \delta  - 1} \right)}}}}} }}} \right)}^{{n_k}}}}
\end{equation}
with $\sum\limits_{{n_{\rm{1}}}, \cdots ,{n_K} = 0}^\infty  {{W_{\bf{n}}}} {\rm{ = }}1$, and ${F_{{\bf A_{\bf{n}}}}}\left( {{z_1}, \cdots ,{z_K}} \right)$ is explicitly expressed as
\begin{equation}\label{eqn:F_ind_snr_CDF}
{F_{{\bf A_{\bf{n}}}}}\left( {{z_1}, \cdots ,{z_K}} \right) = \prod\limits_{k = 1}^K {\frac{{\Upsilon \left( {{n_k} + 1,\frac{{{z_k}}}{{{P_k}{\sigma _k}^2\left( {1 - {\rho ^{2\left( {k + \delta  - 1} \right)}}} \right)}}} \right)}}{{{n_k}!}}}
\end{equation}
where ${\Upsilon \left( \cdot \right)}$ is incomplete Gamma function.
\end{theorem}
\begin{proof}
Please see Appendix \ref{app:proof_of_the_cdf}.
\end{proof}
Clearly from Theorem \ref{the:joint_CDF_mix}, the joint distribution of correlated RVs $\gamma_1,\cdots ,\gamma_K$ can be expressed as a mixture of $K$ independent Gamma RVs. Hereby, by putting (\ref{eqn:joint_CDF_of_snrs_the}) into (\ref{eqn:out_k_harq_fur}), the outage probability ${\mathcal P_{out}}\left( K \right)$ is given by
\begin{equation}\label{eqn:out_prob_fina}
{{\cal P}_{out}}\left( K \right) = \sum\limits_{{n_{\rm{1}}}, \cdots ,{n_K} = 0}^\infty  {{W_{\bf{n}}}{F_{{{\bf{A}}_{\bf{n}}}}}\left( {{2^R} - 1, \cdots ,{2^R} - 1} \right)}
\end{equation}

\subsection{Computation of ${\mathcal P_{out}}\left( K \right)$}
As shown in (\ref{eqn:out_prob_fina}), the outage probability is expressed as the sum of infinite series. It is hard to compute in practice. Here we propose an efficient truncation method to compute the outage probability with high accuracy. After truncation, the outage probability can be written as
%
\begin{equation}\label{eqn:out_prob_truncated}
{\tilde{\cal P}_{out}}\left( K \right) = \sum\limits_{t = 0}^N {\sum\limits_{{n_{\rm{1}}} +  \cdots  + {n_K} = t} {{W_{\bf{n}}}{F_{{{\bf{A}}_{\bf{n}}}}}\left( {{2^R} - 1, \cdots ,{2^R} - 1} \right)} }
\end{equation}
where $N$ is the truncation order. It follows the truncation error as
\begin{align}\label{eqn:trun_erro_def}
\varepsilon  &= {{\cal P}_{out}}\left( K \right) - {\tilde P_{out}}\left( K \right) \notag \\ &= \sum\limits_{t = N + 1}^\infty  {\sum\limits_{{n_{\rm{1}}} +  \cdots  + {n_K} = t} {{W_{\bf{n}}}{F_{{{\bf{A}}_{\bf{n}}}}}\left( {{2^R} - 1, \cdots ,{2^R} - 1} \right)} }.
\end{align}
It is upper bounded by
\begin{equation}\label{eqn:truncate_err_upper}
\varepsilon  \le \mathop {\sup }\limits_{\sum\limits_{k = 1}^K {{n_k}}  > N} \left( {{F_{{{\bf{A}}_{\bf{n}}}}}\left( {{2^R} - 1, \cdots ,{2^R} - 1} \right)} \right) \sum\limits_{t = N + 1}^\infty  {\sum\limits_{{n_{\rm{1}}} +  \cdots  + {n_K} = t} {{W_{\bf{n}}}} }
\end{equation}
Since ${F_{{{\bf{A}}_{\bf{n}}}}}\left( {{2^R} - 1, \cdots ,{2^R} - 1} \right) \le 1$, putting (\ref{eqn:W_n_def}) into (\ref{eqn:truncate_err_upper}) gives
\begin{align}\label{eqn:trun_err_upp_fur}
\varepsilon  &\le \frac{{\rm{1}}}{{1 + \sum\limits_{k = 1}^K {\frac{{{\rho ^{2\left( {k + \delta  - 1} \right)}}}}{{1 - {\rho ^{2\left( {k + \delta  - 1} \right)}}}}} }} \times \notag \\
& \sum\limits_{t = N + 1}^\infty  {\sum\limits_{{n_{\rm{1}}} +  \cdots  + {n_K} = t} {\frac{{t!}}{{\prod\limits_{k = 1}^K {{n_k}!} }}\prod\limits_{k = 1}^K {{{\left( {\frac{{\frac{{{\rho ^{2\left( {k + \delta  - 1} \right)}}}}{{1 - {\rho ^{2\left( {k + \delta  - 1} \right)}}}}}}{{1 + \sum\limits_{k = 1}^K {\frac{{{\rho ^{2\left( {k + \delta  - 1} \right)}}}}{{1 - {\rho ^{2\left( {k + \delta  - 1} \right)}}}}} }}} \right)}^{{n_k}}}} } } \notag\\
 &= \frac{{\rm{1}}}{{1 + \sum\limits_{k = 1}^K {\frac{{{\rho ^{2\left( {k + \delta  - 1} \right)}}}}{{1 - {\rho ^{2\left( {k + \delta  - 1} \right)}}}}} }}\sum\limits_{t = N + 1}^\infty  {{{\left( {\sum\limits_{k = 1}^K {\frac{{\frac{{{\rho ^{2\left( {k + \delta  - 1} \right)}}}}{{1 - {\rho ^{2\left( {k + \delta  - 1} \right)}}}}}}{{1 + \sum\limits_{k = 1}^K {\frac{{{\rho ^{2\left( {k + \delta  - 1} \right)}}}}{{1 - {\rho ^{2\left( {k + \delta  - 1} \right)}}}}} }}} } \right)}^t}} \notag \\
 &= {\left( {\frac{{\sum\limits_{k = 1}^K {\frac{{{\rho ^{2\left( {k + \delta  - 1} \right)}}}}{{1 - {\rho ^{2\left( {k + \delta  - 1} \right)}}}}} }}{{1 + \sum\limits_{k = 1}^K {\frac{{{\rho ^{2\left( {k + \delta  - 1} \right)}}}}{{1 - {\rho ^{2\left( {k + \delta  - 1} \right)}}}}} }}} \right)^{N + 1}}
\end{align}
where the first equality holds by using multinomial expansion. From (\ref{eqn:trun_err_upp_fur}), it is found that the upper bound of the truncation error decreases exponentially as $N$ increases, which demonstrates the effectiveness of the truncation method.

\subsection{Asymptotic analysis}
With the complicated expressions in (\ref{eqn:out_prob_fina}) and (\ref{eqn:out_prob_truncated}), little insight on the outage probability could be found. To better investigate the system behavior, asymptotic analysis under high SNR regime is thus of great importance to extract meaningful insight. To facilitate the analysis, we define $P_k = p_k P_T$. Under high SNR regime $P_T \to \infty$, asymptotic outage probability, diversity order and the impact of time correlation will be studied in the following.
\subsubsection{Asymptotic Outage Probability}
With (\ref{eqn:out_prob_fina}), the outage probability can be written as
\begin{multline}\label{eqn:out_prob_to_der_asy}
{{\cal P}_{out}}\left( K \right) = {W_{\bf{0}}}{F_{{{\bf{A}}_{\bf{0}}}}}\left( {{2^R} - 1, \cdots ,{2^R} - 1} \right) \\
\times \left( {1 + \frac{1}{{{W_{\bf{0}}}}}\sum\limits_{{\bf{n}} \ne {\bf{0}}} {{W_{\bf{n}}}\frac{{{F_{{{\bf{A}}_{\bf{n}}}}}\left( {{2^R} - 1, \cdots ,{2^R} - 1} \right)}}{{{F_{{{\bf{A}}_{\bf{0}}}}}\left( {{2^R} - 1, \cdots ,{2^R} - 1} \right)}}} } \right)
\end{multline}
From the definition of ${F_{{{\bf{A}}_{\bf{n}}}}}\left( {{z_1}, \cdots ,{z_K}} \right)$ in (\ref{eqn:F_ind_snr_CDF}), we can find a special property of ${F_{{{\bf{A}}_{\bf{n}}}}}\left( {{z_1}, \cdots ,{z_K}} \right)$ as follows.
\begin{lemma} \label{lem:cdf_An_asym}
The coefficient ${F_{{{\bf{A}}_{\bf{n}}}}}\left( {{z_1}, \cdots ,{z_K}} \right)$ can be written as
\begin{multline}\label{eqn:lem_cdf_ind_fin_asy}
{F_{{{\bf{A}}_{\bf{n}}}}}\left( {{z_1}, \cdots ,{z_K}} \right) = {P_T}^{-\sum\limits_{k = 1}^K {\left( {{n_k} + 1} \right)} } \times \\
\prod\limits_{k = 1}^K {\frac{{{{\left( {\frac{{{z_k}}}{{{p_k}{\sigma _k}^2\left( {1 - {\rho ^{2\left( {k + \delta  - 1} \right)}}} \right)}}} \right)}^{{n_k} + 1}}}}{{{n_k}!\left( {{n_k} + 1} \right)}}}  + o\left( {{P_T}^{-\sum\limits_{k = 1}^K {\left( {{n_k} + 1} \right)} }} \right)
\end{multline}
and it satisfies
\begin{align}\label{eqn:lemma_ratio_two_cdf_ind1}
\frac{{{F_{{{\bf{A}}_{\bf{n}}}}}\left( {{2^R} - 1, \cdots ,{2^R} - 1} \right)}}{{{F_{{{\bf{A}}_{\bf{0}}}}}\left( {{2^R} - 1, \cdots ,{2^R} - 1} \right)}} = o({P_T}^{-0.5}), \, {\bf n \ne 0},
\end{align}
where $o(\cdot)$ denotes higher-order infinitesimal.
\end{lemma}
\begin{proof}
Please see Appendix \ref{app:proof_lemma_2}.
\end{proof}

By using the result in Lemma \ref{lem:cdf_An_asym}, the outage probability(\ref{eqn:out_prob_to_der_asy}) under high SNR can be approximated as
\begin{align}
\label{eqn:out_prob_asym_ind}
{{\cal P}_{out}}\left( K \right) &= {W_{\bf{0}}}{F_{{{\bf{A}}_{\bf{0}}}}}\left( {{2^R} - 1, \cdots ,{2^R} - 1} \right)\left( {1 + o({P_T}^{ - 0.5})} \right) \nonumber \\
&\approx {W_{\bf{0}}}{F_{{{\bf{A}}_{\bf{0}}}}}\left( {{2^R} - 1, \cdots ,{2^R} - 1} \right)
\end{align}
Clearly from (\ref{eqn:out_prob_asym_ind}), under high SNR regime, the joint CDF of correlated RVs $\gamma_1,\cdots ,\gamma_K$ can be approximated as a weighted joint CDF of independent Gamma RVs with parameters $(1,{{P_k}{\sigma _k}^2\left( {1 - {\rho ^{2\left( {k + \delta  - 1} \right)}}} \right)})$. With this result, optimal system design for Type I HARQ can be simplified and meaningful insight can be extracted.

By plugging (\ref{eqn:lem_cdf_ind_fin_asy}) into (\ref{eqn:out_prob_asym_ind}), it follows that
\begin{multline}\label{eqn:asym_out_fin}
{{\cal P}_{out}}\left( K \right)
= {\frac{{{W_{\bf{0}}}{{\left( {{2^R} - 1} \right)}^K}}}{{{P_T}^K}}\prod\limits_{k = 1}^K {\frac{1}{{{p_k}{\sigma _k}^2\left( {1 - {\rho ^{2\left( {k + \delta  - 1} \right)}}} \right)}}} } \\
+ {o\left( {{P_T}^{ - K}} \right)}
\end{multline}
By using (\ref{eqn:W_n_def}), the asymptotic outage probability can be finally written as
\begin{equation}\label{eqn:out_prob_fina_phys}
{{\cal P}_{out}}\left( K \right) \approx \underbrace {{{\left( {{2^R} - 1} \right)}^K}}_{A}\underbrace {\prod\limits_{k = 1}^K {\frac{1}{{{P_k}{\sigma _k}^2}}} }_{B}\underbrace {\frac{{\rm{1}}}{{\ell \left( \rho,K  \right)}}}_{C}
\end{equation}
where $\ell \left( \rho,K  \right) = \left( {1 + \sum\limits_{k = 1}^K {\frac{{{\rho ^{2\left( {k + \delta  - 1} \right)}}}}{{1 - {\rho ^{2\left( {k + \delta  - 1} \right)}}}}} } \right)\prod\limits_{k = 1}^K {\left( {1 - {\rho ^{2\left( {k + \delta  - 1} \right)}}} \right)} $. With (\ref{eqn:out_prob_fina_phys}), the effects of coding rate, transmit powers and time correlation can now be clearly seen from the individual terms A, B and C, respectively.

\subsubsection{Diversity Order}
 The diversity order $d$ is defined as \cite{chelli2014performance,zheng2003diversity}
\begin{equation}\label{eqn:diver_order_def}
d =  - \mathop {\lim }\limits_{{P_T} \to \infty } \frac{{\ln \left( {{{\cal P}_{out}}\left( K \right)} \right)}}{{\ln \left( {{P_T}} \right)}}
\end{equation}
By using (\ref{eqn:asym_out_fin}), it follows that
\begin{align}\label{eqn:diversity_order_def_der}
d &= - \mathop {\lim }\limits_{{P_T} \to \infty } \frac{{\ln ( {\frac{{{W_{\bf{0}}}{{\left( {{2^R} - 1} \right)}^K}}}{{{P_T}^K}}\prod\limits_{k = 1}^K {\frac{1}{{{p_k}{\sigma _k}^2\left( {1 - {\rho ^{2\left( {k + \delta  - 1} \right)}}} \right)}}}  + o\left( {{P_T}^{ - K}} \right)} )}}{{\ln \left( {{P_T}} \right)}} \notag\\
 &=  - \mathop {\lim }\limits_{{P_T} \to \infty } \frac{{\ln \left( {{P_T}^{ - K}} \right)}}{{\ln \left( {{P_T}} \right)}} - \mathop {\lim }\limits_{{P_T} \to \infty } \frac{{\ln \left( {1 + \frac{{o\left( {{P_T}^{ - K}} \right)}}{{{P_T}^{ - K}}}} \right)}}{{\ln \left( {{P_T}} \right)}} \notag\\
&= K - \mathop {\lim }\limits_{{P_T} \to \infty } \frac{{o\left( {{P_T}^{ - K}} \right)}}{{{P_T}^{ - K}\ln \left( {{P_T}} \right)}} = K, \, \rho \ne 1,
\end{align}
The third equality holds by using the equivalent infinitesimals as $\ln \left( {1 + \frac{{o\left( {{P_T}^{ - K}} \right)}}{{{P_T}^{ - K}}}} \right) \sim \frac{{o\left( {{P_T}^{ - K}} \right)}}{{{P_T}^{ - K}}}$. It is thus proved that the diversity order of Type I HARQ is equal to the number of transmissions $K$, i.e., full diversity can be achieved even under time-correlated fading channels when $\rho \ne 1$. It is worth noting that the conclusion of full diversity does not hold in the case of fully correlated fading channels, i.e., $\rho=1$. Under fully correlated fading channels ($\rho=1$), no time diversity can be achieved from retransmissions and the diversity order reduces to $1$.

\subsubsection{Impact of Time Correlation}
From (\ref{eqn:out_prob_fina_phys}), the impact of time correlation on outage probability under high SNR regime can be further analyzed from ${\ell \left( \rho,K  \right)}$. The result is shown in the following lemma.
\begin{lemma}\label{lem:lemma_3}
${\ell \left( \rho,K  \right)}$ is a decreasing function with respect to the time correlation coefficient $\rho$. Specifically, ${\ell \left( \rho,K  \right)} \le {\ell \left( 0,K  \right)} = 1$.
\end{lemma}
\begin{proof}
Please see Appendix \ref{app:proof_lemma}.
\end{proof}
Lemma \ref{lem:lemma_3} reveals that the the presence of time correlation will degrade the system performance under high SNR regime, that is, it causes the increase of outage probability.

\section{Numerical Results and Discussions} \label{sec:num}
In this section, numerical results are shown to test the accuracy of our outage analysis. In the following, we take systems with $\sigma_1=\cdots=\sigma_K=1$ and $R=2~\rm bps/Hz$ as examples.

In Fig. \ref{fig:1}, the outage probability is plotted against transmit power $P_T$ with $\rho=0.5$ and $N=5$. There is a perfect match between Monte Carlo simulation results and analytical results, which demonstrates the correctness of our analysis. Under high SNR regime, four curves coincide well with each other. In addition, it can be readily found that the diversity order is equal to the number of transmissions. For example, for $K=4$, as the transmit power $P_T$ increases from $20\rm dB$ to $30\rm dB$, the outage probability reduces from $10^{-6}$ to $10^{-10}$. Thus the diversity order is $4$. Furthermore, it can be observed that the increase of the number of transmissions $K$ will cause a significant reduction of outage probability.
\begin{figure}
  \centering
  \includegraphics[width=3.5in]{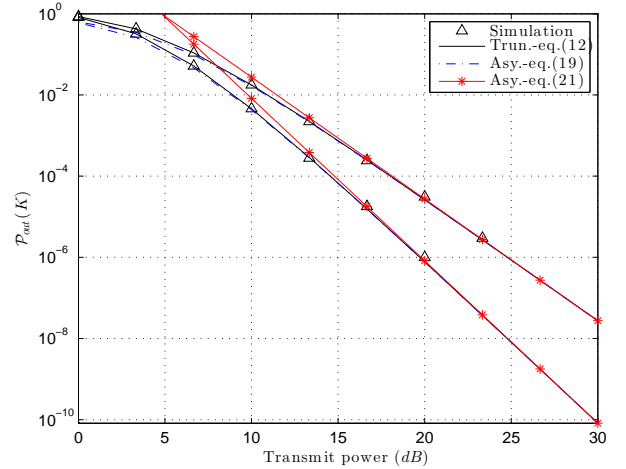}\\
  \caption{Outage probability versus transmit power $P_T$.}\label{fig:1}
\end{figure}

To show the effect of truncation order $N$, the approximated outage probability after truncation $\tilde{\mathcal P}_{out}(K)$ is plotted versus $N$ with $K=4$ in Fig. \ref{fig:out_asy}. It is readily found that truncation order of $N=5$ is enough to well approximate ${\mathcal P}_{out}(K)$ with negligible error. In addition, low truncation order $N$ is sufficient to achieve a good approximation of ${\mathcal P}_{out}(K)$ under high SNR regime or low $\rho$. For example, the truncation order of $N=2$ can achieve a good approximation when $\rho=0.5$ or $P_T=10\rm dB$.
\begin{figure}
  \centering
  \includegraphics[width=3.5in]{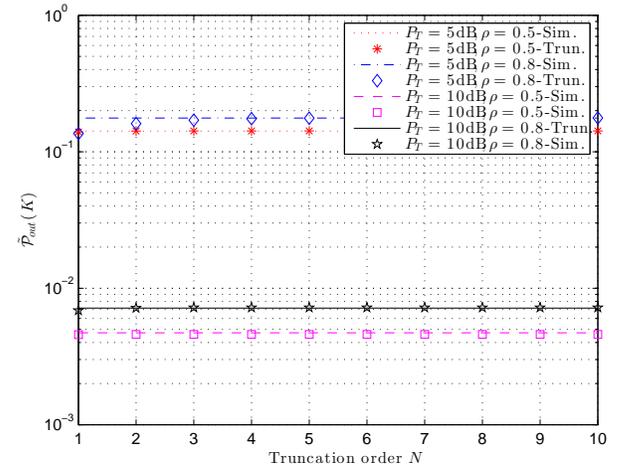}\\
  \caption{Effect of truncation order $N$.}\label{fig:out_asy}
\end{figure}

Fig. \ref{fig:diver} shows the impact of time correlation on Type I HARQ. It is readily observed that $\ell (\rho,K)$ decreases with $\rho$, which reveals that time correlation has negative effect on outage probability under high SNR regime.
\begin{figure}
  \centering
  \includegraphics[width=3.5in]{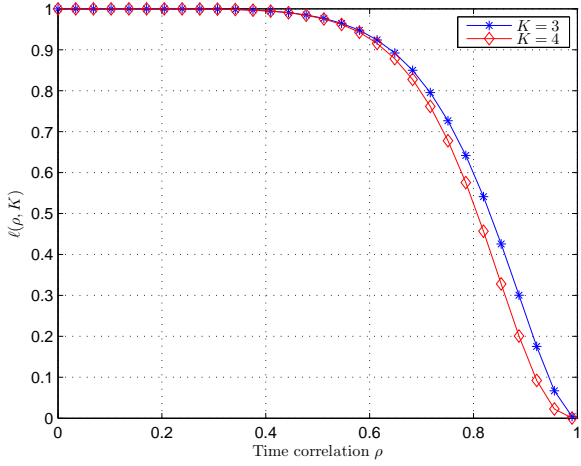}\\
  \caption{Impact of time correlation.}\label{fig:diver}
\end{figure}

\section{Conclusions}\label{sec:con}
Type I HARQ over time correlated Rayleigh fading channels has been particularly analyzed in this paper. Its outage probability has been derived as a weighted sum of joint CDF of independent Gamma RV, which enables an efficient calculation of the outage probability based on truncation. Then asymptotic analysis has also been conducted to extract meaningful insight of various parameters. It has revealed that full diversity can be achieved and time correlation has detrimental impact on system performance.
\appendices
\section{Proof of Theorem \ref{the:joint_CDF_mix}} \label{app:proof_of_the_cdf}
The joint CDF ${F_{{\gamma _1}, \cdots ,{\gamma _K}}}\left( {{z_1}, \cdots ,{z_K}} \right)$ can be written as
\begin{multline}\label{eqn:joint_cdf_der}
{F_{{\gamma _1}, \cdots ,{\gamma _K}}}\left( {{z_1}, \cdots ,{z_K}} \right) \\
  = \int\limits_0^{{z_1}} { \cdots \int\limits_0^{{z_K}} {{f_{{\gamma _1}, \cdots ,{\gamma _K}}}\left( {{r_1}, \cdots ,{r_K}} \right)d{r_1} \cdots d{r_k}} }
\end{multline}
It follows from Lemma \ref{lemm:pdf} that
\begin{multline}\label{eqn:cdf_der_from_lemma_1}
{F_{{\gamma _1}, \cdots ,{\gamma _K}}}\left( {{z_1}, \cdots ,{z_K}} \right) = \int\limits_0^\infty  {{e^{ - t}}\prod\limits_{k = 1}^K {\int\limits_0^{{z_k}} {\frac{1}{{{P_k}{\sigma _k}^2\left( {1 - {\rho ^{2\left( {k + \delta  - 1} \right)}}} \right)}}} } } \\
\times {e^{ - \frac{{{r_k} + {P_k}{\sigma _k}^2{\rho ^{2\left( {k + \delta  - 1} \right)}}t}}{{{P_k}{\sigma _k}^2\left( {1 - {\rho ^{2\left( {k + \delta  - 1} \right)}}} \right)}}}}_0{F_1}\left( {;1;\frac{{{\rho ^{2\left( {k + \delta  - 1} \right)}}t{r_k}}}{{{P_k}{\sigma _k}^2{{\left( {1 - {\rho ^{2\left( {k + \delta  - 1} \right)}}} \right)}^2}}}} \right)d{r_k}dt
\end{multline}
By using the series representation of hypergeometric function \cite[Eq. 1.116]{mathai2009h}, (\ref{eqn:cdf_der_from_lemma_1}) can be derived as
\begin{multline}\label{eqn:cdf_der_serie_exp}
{F_{{\gamma _1}, \cdots ,{\gamma _K}}}\left( {{z_1}, \cdots ,{z_K}} \right) = \\
\prod\limits_{k = 1}^K {\frac{1}{{{P_k}{\sigma _k}^2\left( {1 - {\rho ^{2\left( {k + \delta  - 1} \right)}}} \right)}}} \int\limits_0^\infty  {{e^{ - \left( {1 + \sum\limits_{k = 1}^K {\frac{{{\rho ^{2\left( {k + \delta  - 1} \right)}}}}{{1 - {\rho ^{2\left( {k + \delta  - 1} \right)}}}}} } \right)t}}}  \times \\
\prod\limits_{k = 1}^K {\int\limits_0^{{z_k}} {{e^{ - \frac{{{r_k}}}{{{P_k}{\sigma _k}^2\left( {1 - {\rho ^{2\left( {k + \delta  - 1} \right)}}} \right)}}}}\sum\limits_{{n_k} = 0}^\infty  {\frac{{{{\left( {\frac{{{\rho ^{2\left( {k + \delta  - 1} \right)}}t{r_k}}}{{{P_k}{\sigma _k}^2{{\left( {1 - {\rho ^{2\left( {k + \delta  - 1} \right)}}} \right)}^2}}}} \right)}^{{n_k}}}}}{{{{\left( {{n_k}!} \right)}^2}}}} d{r_k}} } dt
\end{multline}
By exchanging the order of summation and multiplication, it yields
\begin{multline}\label{eqn:cdf_exchan_order_give}
{F_{{\gamma _1}, \cdots ,{\gamma _K}}}\left( {{z_1}, \cdots ,{z_K}} \right) = \prod\limits_{k = 1}^K {\frac{1}{{{P_k}{\sigma _k}^2\left( {1 - {\rho ^{2\left( {k + \delta  - 1} \right)}}} \right)}}}   \\
\times \sum\limits_{{n_{\rm{1}}}, \cdots ,{n_K} = 0}^\infty  {\prod\limits_{k = 1}^K {\frac{{{{\left( {\frac{{{\rho ^{2\left( {k + \delta  - 1} \right)}}}}{{{P_k}{\sigma _k}^2{{\left( {1 - {\rho ^{2\left( {k + \delta  - 1} \right)}}} \right)}^2}}}} \right)}^{{n_k}}}}}{{{{\left( {{n_k}!} \right)}^2}}}} } \\
\times \int\limits_0^\infty  {{t^{\sum\limits_{k = 1}^K {{n_k}} }}{e^{ - \left( {1 + \sum\limits_{k = 1}^K {\frac{{{\rho ^{2\left( {k + \delta  - 1} \right)}}}}{{1 - {\rho ^{2\left( {k + \delta  - 1} \right)}}}}} } \right)t}}dt} \\
\times \prod\limits_{k = 1}^K {\int\limits_0^{{z_k}} {{r_k}^{{n_k}}{e^{ - \frac{{{r_k}}}{{{P_k}{\sigma _k}^2\left( {1 - {\rho ^{2\left( {k + \delta  - 1} \right)}}} \right)}}}}d{r_k}} }
\end{multline}
By using \cite[Eqs. 3.381.1, and 3.381.4]{gradshteyn1965table} and conducting some algebraic manipulations, Theorem \ref{the:joint_CDF_mix} directly follows. Moreover, by taking limits of (\ref{eqn:joint_CDF_of_snrs_the}) as $z_1,\cdots,z_K \to \infty$, and using $\mathop {\lim }\limits_{{z_1}, \cdots ,{z_K} \to \infty } {F_{{\gamma _1}, \cdots ,{\gamma _k}}}\left( {{z_1}, \cdots ,{z_K}} \right) = \mathop {\lim }\limits_{{z_1}, \cdots ,{z_K} \to \infty } {F_{{\bf A_{\bf{n}}}}}\left( {{z_1}, \cdots ,{z_K}} \right) = 1$, $\sum\limits_{{n_{\rm{1}}}, \cdots ,{n_K} = 0}^\infty  {{W_{\bf{n}}}} {\rm{ = }}1$ holds without dispute.

\section{Proof of Lemma \ref{lem:cdf_An_asym}}\label{app:proof_lemma_2}
By using \cite[Eq. 8.354.1]{gradshteyn1965table}, (\ref{eqn:F_ind_snr_CDF}) can be further written as
\begin{multline}\label{eqn:F_ind_snr_cdf_fur}
{F_{{{\bf{A}}_{\bf{n}}}}}\left( {{z_1}, \cdots ,{z_K}} \right) = \prod\limits_{k = 1}^K   {\frac{1}{{{n_k}!}}} \times\\
 \sum\limits_{{m_k} = 0}^\infty  {\frac{{{{\left( { - 1} \right)}^{{m_k}}}{{\left( {\frac{{{z_k}}}{{{p_k}{P_T}{\sigma _k}^2\left( {1 - {\rho ^{2\left( {k + \delta  - 1} \right)}}} \right)}}} \right)}^{{n_k} + {m_k} + 1}}}}{{{m_k}!\left( {{n_k} + {m_k} + 1} \right)}}}
\end{multline}
Under high SNR regime, it follows from (\ref{eqn:F_ind_snr_cdf_fur}) that
\begin{multline}\label{eqn:F_ind_snr_cdf_fur2}
{F_{{{\bf{A}}_{\bf{n}}}}}\left( {{z_1}, \cdots ,{z_K}} \right) = {P_T}^{ - \sum\limits_{k = 1}^K {\left( {{n_k} + 1} \right)} }\prod\limits_{k = 1}^K {\frac{{{{\left( {\frac{{{z_k}}}{{{p_k}{\sigma _k}^2\left( {1 - {\rho ^{2\left( {k + \delta  - 1} \right)}}} \right)}}} \right)}^{{n_k} + 1}}}}{{{n_k}!\left( {{n_k} + 1} \right)}}}  \\
 + o\left( {{P_T}^{ - \sum\limits_{k = 1}^K {\left( {{n_k} + 1} \right)} }} \right)
\end{multline}
where the notation $o(x)$ defines high order infinitesimal of $x$, i.e., the ratio $o(x)/x$ approaches to zero as $x \to 0$. Thus we have
\begin{multline}\label{eqn:ratio_two_cdf_ind}
\frac{{{F_{{{\bf{A}}_{\bf{n}}}}}\left( {{2^R} - 1, \cdots ,{2^R} - 1} \right)}}{{{F_{{{\bf{A}}_{\bf{0}}}}}\left( {{2^R} - 1, \cdots ,{2^R} - 1} \right)}} = {P_T}^{ - \sum\limits_{k = 1}^K {{n_k}} } \times \\
\prod\limits_{k = 1}^K {\frac{{{{\left( {\frac{{{z_k}}}{{{p_k}{\sigma _k}^2\left( {1 - {\rho ^{2\left( {k + \delta  - 1} \right)}}} \right)}}} \right)}^{{n_k}}}}}{{{n_k}!\left( {{n_k} + 1} \right)}}}  + o\left( {{P_T}^{ - \sum\limits_{k = 1}^K {{n_k}} }} \right)
\end{multline}
When $\bf n \ne 0$, $\sum\limits_{k = 1}^K {{n_k}} \ge 1$ and it follows that
\begin{align}\label{eqn:ratio_two_cdf_ind1}
\frac{{{F_{{{\bf{A}}_{\bf{n}}}}}\left( {{2^R} - 1, \cdots ,{2^R} - 1} \right)}}{{{F_{{{\bf{A}}_{\bf{0}}}}}\left( {{2^R} - 1, \cdots ,{2^R} - 1} \right)}} = o({P_T}^{-\kappa}),\, {\bf n \ne 0}, \kappa  < \sum\limits_{k = 1}^K {{n_k}} .
\end{align}
Thus the lemma holds in the case of $\kappa=0.5$.

\section{Proof of Lemma \ref{lem:lemma_3}} \label{app:proof_lemma}
To prove the monotonically increasing of ${\ell \left( \rho,K  \right)}$ with respect to $\rho$, we assume ${\Delta \rho } > 0$. Then it follows from the definition that
\begin{align}\label{eqn:rho_prov1}
\ell \left( {\rho  + \Delta \rho ,K} \right) &= \left( {1 + \sum\limits_{k = 2}^K {\frac{{{{\left( {\rho  + \Delta \rho } \right)}^{2\left( {k + \delta  - 1} \right)}}\left( {1 - {{\left( {\rho  + \Delta \rho } \right)}^{2\delta }}} \right)}}{{1 - {{\left( {\rho  + \Delta \rho } \right)}^{2\left( {k + \delta  - 1} \right)}}}}} } \right) \notag \\
&\times \prod\limits_{k = 2}^K {\left( {1 - {{\left( {\rho  + \Delta \rho } \right)}^{2\left( {k + \delta  - 1} \right)}}} \right)}
\end{align}
Since ${\Delta \rho } > 0$, the following inequality holds
\begin{align}\label{eqn:rho_prov2}
\ell \left( {\rho  + \Delta \rho ,K} \right) &< \left( {1 + \sum\limits_{k = 2}^K {\frac{{{{\left( {\rho  + \Delta \rho } \right)}^{2\left( {k + \delta  - 1} \right)}}\left( {1 - {\rho ^{2\delta }}} \right)}}{{1 - {{\left( {\rho  + \Delta \rho } \right)}^{2\left( {k + \delta  - 1} \right)}}}}} } \right)\notag \\
&\prod\limits_{k = 2}^K {\left( {1 - {{\left( {\rho  + \Delta \rho } \right)}^{2\left( {k + \delta  - 1} \right)}}} \right)}
\end{align}
Finally, (\ref{eqn:rho_prov2}) can be rewritten as
\begin{multline}\label{eqn:rho_prov}
\ell \left( {\rho  + \Delta \rho ,K} \right) < \left( {1 + \frac{{{\rho ^{2\delta }}}}{{1 - {\rho ^{2\delta }}}} + \sum\limits_{k = 2}^K {\frac{{{{\left( {\rho  + \Delta \rho } \right)}^{2\left( {k + \delta  - 1} \right)}}}}{{1 - {{\left( {\rho  + \Delta \rho } \right)}^{2\left( {k + \delta  - 1} \right)}}}}} } \right)  \\
\times \left( {1 - {\rho ^{2\delta }}} \right) \prod\limits_{k = 2}^K {\left( {1 - {{\left( {\rho  + \Delta \rho } \right)}^{2\left( {k + \delta  - 1} \right)}}} \right)}
\end{multline}
Following the same procedure as (\ref{eqn:rho_prov1}) - (\ref{eqn:rho_prov}), we can prove that
\begin{equation}\label{eqn:ell_prov_rel}
\ell \left( {\rho  + \Delta \rho,K } \right) < \ell \left( \rho,K  \right)
\end{equation}
\bibliographystyle{ieeetran}
\bibliography{manuscript_1}
\end{document}